\documentclass[twocolumn,twoside]{IEEEtran}

\let\svthefootnote\thefootnote

\usepackage{amssymb}
\usepackage{amsfonts}
\usepackage{amsmath}
\usepackage{amsthm}
\usepackage[all]{xy}
\usepackage{latexsym}
\usepackage{enumerate}
\usepackage{algorithm}

\DeclareMathOperator{\tr}{Tr}

\DeclareMathOperator{\im}{Im}

\DeclareMathOperator{\linspan}{Span}

\DeclareMathOperator{\aut}{Aut}
\DeclareMathOperator{\rank}{rank}

\DeclareMathOperator{\diag}{diag}

\newcommand{\deff}{\mbox{$\stackrel{\rm def}{=}$}}

\newcommand{\ef}{\mathbb{F}}

\newcommand{\efq}{\ef_q}
\newcommand{\efqm}{\ef_{q^m}}

\newcommand{\rs}{\mathsf{RS}}
\newcommand{\grs}{\mathsf{GRS}}

\newcommand{\bad}{\mathsf{Bad}}
\newcommand{\reps}{\mathsf{Reps}}
\newcommand{\prob}{\mathrm{Prob}}
\newcommand{\fim}{f^{\im}}

\newcommand{\undermat}[2]{%
  \makebox[0pt][l]{$\smash{\underbrace{\phantom{%
    \begin{matrix}#2\end{matrix}}}_{\text{$#1$}}}$}#2}

\newtheorem{proposition}{Proposition}
\newtheorem{definition}[proposition]{Definition}
\newtheorem{corollary}[proposition]{Corollary}
\newtheorem{lemma}[proposition]{Lemma}
\newtheorem{theorem}[proposition]{Theorem}
\newtheorem{remark}[proposition]{Remark}

\begin{document}
\title{Repairing Reed--Solomon Codes Evaluated on Subspaces}

\author{\textbf{Amit Berman}\IEEEauthorrefmark{1}, \textbf{Sarit Buzaglo}\IEEEauthorrefmark{1}, \textbf{Avner Dor}\IEEEauthorrefmark{1}, \textbf{Yaron Shany}\IEEEauthorrefmark{1}, and \IEEEauthorblockN{\textbf{Itzhak Tamo}\IEEEauthorrefmark{2}\IEEEauthorrefmark{1}}

\IEEEauthorblockA{\IEEEauthorrefmark{1}Samsung Semiconductor Israel R\&D Center, 2 Shoham St., Ramat Gan, 5251003, Israel\\}

\IEEEauthorblockA{\IEEEauthorrefmark{2}Department of Electrical Engineering-Systems, Tel Aviv University, Tel Aviv 6997801, Israel\\}

{\it  \{amit.berman, sarit.b, avner.dor, yaron.shany\}@samsung.com}, {\it zactamo@gmail.com \vspace{-5ex}}}

\maketitle\let\thefootnote\relax\footnote{This work was carried out at Samsung Semiconductor Israel R\&D Center.}
\addtocounter{footnote}{-1}\let\thefootnote\svthefootnote

\begin{abstract}
We consider the repair problem for Reed--Solomon (RS)
codes, evaluated on an $\efq$-linear subspace $U\subseteq\efqm$ of dimension $d$, where $q$ is a prime
power, $m$ is a positive integer, and $\efq$ is the Galois field of size $q$. For the case of
$q\geq 3$, we show the existence of a linear repair scheme for the RS code of length $n=q^d$ and codimension $q^s$, $s< d$,
evaluated on $U$, in
which each of the $n-1$ surviving nodes transmits only $r$ symbols of
$\efq$, provided that $ms\geq d(m-r)$.  For the case of
$q=2$, we prove a similar result, with some restrictions
on the evaluation linear subspace $U$. Our proof is based on a
probabilistic argument, however the result is not merely an existence result; the success
probability is fairly large (at least $1/3$) and there is a simple
criterion for checking the validity of the randomly chosen linear repair scheme.
Our result extend the construction of Dau--Milenkovich to the range $r<m-s$,
for a wide range of parameters.
\end{abstract}

\section{Introduction}

Erasure codes are widely used for increasing the reliability of
distributed storage systems. In such systems, data is encoded and
stored on several {\it nodes}, where each storage node corresponds to
one coordinate of the erasure code. To minimize storage overhead due
to coding, erasure codes used in practice are typically Maximum
Distance Separable (MDS) codes. While an erasure code can typically
recover from several node failures (i.e., from more than a single
erasure), a single-node failure is the most common type of
failure~\cite{RSGKB}. Hence, there is an interest in finding MDS codes that can
efficiently repair a single node failure.

To repair a single node failure, the system has to download part of the content of
some of the surviving nodes, called {\it helper nodes}. The total amount of data downloaded from the
helper nodes is called the {\it repair bandwidth}. A code designed for
minimizing the repair bandwidth is called a {\it regenerating
code}. Regenerating codes have been studied extensively since the
introduction of the subject in \cite{DGWWR}. A convenient way to measure the repair bandwidth is through the concept of {\it sub-packetization}, where data is divided to smaller units of a fixed size and each helper node transmits some function of these units.
A common approach, which is adopted in this paper, is to utilize a sub-packetization is to consider codes over the extension field $\efqm$ (over $\efq$), where each data node is composed of $m$ symbols of $\efq$, hence, the units in the sub-packetization are $\efq$-symbols.
A code that is defined over
$\efq^{m}$ is called an
{\it array code} of {\it sub-packetization} $m$, if $m\geq 2$, and is called a {\it scalar code}, otherwise.


For an MDS code of length $n$ and dimension $k$ over $\efqm$, the {\it cut-set bound} \cite{DGWWR} states that the repair bandwidth
is at least $hm/(h+1-k)$ $\efq$-symbols, where $h$ is the maximum number of helper nodes that participates in a single node repair. Thus, the repair bandwidth is minimized when $h$ takes its maximum possible value of $n-1$. In this
paper we consider only the case $h=n-1$, for which the cut-set
bound reads $(n-1)m/(n-k)$. An MDS array code achieving the cut-set
bound is called a {\it minimum storage regenerating (MSR)} code. By
now, there are several constructions of MSR array codes (see, e.g.,~
\cite{TWB} and ~\cite{YB}).


Guruswami and Wootters (GW)~\cite{GW} introduced a useful characterization
of \emph{linear} repair scheme for linear MDS codes in terms of
appropriate codewords of the dual codes. In the same paper, Guruswami
and Wootters also introduced a linear repair scheme for
Reed--Solomon (RS) codes over $\efqm$, of full-length (i.e., their evaluation-set is the entire field) and of codimension
$q^{m-1}$. This linear repair scheme is optimal, that is,
it achieves the minimum possible repair
bandwidth of any linear repair scheme with the same code parameters and sub-packetization. The result of Guruswami and Wootters was later extended by Dau and
Milenkovich (DM)~\cite{DM}, who presented linear repair schemes
for RS codes with higher dimensions, which is optimal only for RS codes.

While the schemes of~\cite{DM} and~\cite{GW} are optimal
for full-length RS codes, where the number of data units $m$ in the sub-packetization is
logarithmic in the length, they are quite far from the cut-set
bound. Until recently, it was an open question whether scalar MDS codes, particularly RS codes, can achieve the cut-set
bound. This question was answered in~\cite{TYB}, where an explicit evaluation set was presented for which the corresponding RS codes achieve the cut-set
bound. For practical implementation, however, this construction is infeasible, since it requires $m$ to be exponential in $n\log n$, where $n$ is the code length~\cite{TYB}. For this reason, there is both a
practical and a theoretical interest to further explore the tradeoff
between the number of data units ($m$) and the repair bandwidth of RS codes and to find
additional repair schemes. This
direction has been recently pursued in~\cite{GJ},~\cite{LWJ}, and~\cite{LWJ2}.

In this paper we consider linear repair schemes for RS
codes evaluated on an $\efq$-linear
subspace $U\subseteq\efqm$, in which each surviving node transmits $r$ $\efq$-symbols
for the repair of the failed node.
When $q$ is greater than two, we show the existence of such a
linear repair scheme for every choice of $U$,
provided that $ms\geq d(m-r)$, where $q^s$ is the codimension of the RS code and $d$ is the dimension of $U$.
For the case of $q=2$, we prove that such a linear
repair scheme exists for every choice of $U$,
whenever $ms\geq d(m-r)+1$, and for many $\efq$-linear
subspaces, when $ms=d(m-r)$.
Our result translates to a practical
probabilistic algorithm that outputs with high probability
a linear repair scheme for the code, since success probability
is fairly large (at least $1/3$) and there is a simple
algorithm for checking the validity of the construction.
Our result generalize the result of Dau and
Milenkovich and the ``scheme in one coset'' presented in~\cite{LWJ} and~\cite{LWJ2}.

A useful property of our scheme is a duality property between the pair of parameters
$d$ and $r$, and the pair of parameters $m-r$ and $m-d$.
Namely, assume $C$ is an RS code that is evaluated on an $\efq$-linear subspace of dimension
$d$, and that our construction generates a linear repair scheme for $C$ in which each surviving node transmits $r$ $\efq$-symbols.
Then there is an explicit way to derive a linear repair scheme for an RS code that is evaluated on an $\efq$-linear subspace of dimension $m-r$, in which each surviving node transmits $m-d$ $\efq$-symbols.

We also present an explicit construction for the special cases
where $d$ or $m-r$ divides $m$ and for a specific choice of $U$.
When $d$ divides $m$, we set $U$ to be the subfield $\ef_{q^d}$,
and present an explicit construction that is almost identical to
the scheme of Li {\it et al.}~\cite{LWJ2}.
Notice that, our existence result supports a much wider parameters range
as $d$ may not divide $m$ and the evaluation set may be
\emph{any} subspace of dimension $d$.
The case that $m-r$ divides $m$ follows immediately
from the duality of our scheme.

The rest of this paper is organized as follows. In Section~\ref{sec:background}
we present some of the basic concepts that are used throughout the paper. In particular,
we recall the concept of a linear repair scheme and review the
important result from~\cite{GW} that provides a
convenient criterion for the existence of a linear repair scheme.
In Section~\ref{sec:lrs_rs}, we present some general results on linear repair schemes
for RS codes that are evaluated on linear subspaces.
The main result of the paper is given in Section~\ref{sec:mainres}.
In Section~\ref{sec:construction} we show explicit constructions, where $d$ or $m-r$ divides $m$.
Some specific examples are given in Section~\ref{sec:eg} and we conclude the paper in Section~\ref{sec:conclusion}.

\section{Preliminaries}\label{sec:background}

The set of all polynomials in the variable $X$ with coefficients taken from a field $\ef$ is denoted by $\ef [X]$. The degree of a polynomial $f\in \ef[X]$ is denoted by $\deg(f)$.
For a subset $S$ of an $\efq$-linear space, the $\efq$-linear subspace that is spanned by $S$
is denoted by $\linspan_{q}(S)$ and the rank of $S$ (the dimension of $\linspan_{q}(S)$) is denoted by $\rank_q(S)$.
For a vector $\mathbf{s}=(s_1,\ldots,s_{\ell})\in \efqm^{\ell}$  we will write $\rank_q(\mathbf{s})$, for $\rank_q(s_1,\ldots,s_{\ell})$.
As usual, for a matrix $A\in \efq^{n\times n}$, the rank of $A$ over $\efq$ is denoted by $\rank_q(A)$. 

Let $V\subseteq \efqm$ be an $\efq$-subspace
of dimension $r$ with a basis $B=\{b_1,\ldots,b_r\}$ and let $S=\{b_1,\ldots,b_m\}$ be a basis for $\efqm$
that contains $B$. For $x\in \efqm$, the {\it projection} of $x$ to $V$, ${x}_{V}$, is the unique element $v\in V$ such that
$x=v+w$, for some (unique) $w\in \linspan_{q}(S\setminus B)$.
For an element $u\in \efqm$, consider the $\efq$-linear map $F_u:\efqm\to \efqm$ defined by $F_u(x)\deff u\cdot x$, and let $[u]_S\in \efq^{m\times m}$ be the
matrix representation of $F_u$ by right multiplication, according to the basis $S$. That is, if $\mathbf{x}=(x_1,x_2,\ldots, x_m)\in \efq^m$ is the vector representation of $x\in \efqm$ according to the basis $S$, then $[u]_S\cdot\mathbf{x}^T$ is the vector representation of $F_u(x)$ according to the basis $S$. We denote by $[u]_{B,S}\in\efq^{r\times m}$ the matrix
consisting of the $r$ rows of $[u]_S$ corresponding to the
elements of $B$. Note that,
right multiplication by $[u]_{B,S}$
represents the linear map that maps
$x$ to the projection of $F_u(x)$ to $V$.
Similarly, we denote by $[u]_{S,B}\in\efq^{m\times r}$ the matrix
consisting of the $r$ columns of $[u]_S$ corresponding to the
elements of $B$. Right multiplication by $[u]_{S,B}$
represents the linear map that projects $x$ to $V$ and multiplies the result by $u$.

As usual, an $[n,k]_q$ code $C$ is a linear code of length $n$ and
dimension $k$, over the field $\ef_q$.
The {\it dual code} of  an $[n,k]_q$ code $C$,
$C^{\ast}\subseteq \efq^n$,
\footnote{We use the superscript ${}^{\ast}$
instead of the conventional notation ${}^{\perp}$
to denote the dual code.
The later is used throughout this paper to denote a
different type of duality that is defined through the
trace map and has a more prominent role in this paper.}
is an $[n,n-k]_q$ code defined by
$$
C^{\ast}\deff\left\{(x_1,\ldots,x_{n})\in \ef_q^n~:~\forall
\mathbf{c}\in C,~\sum_{i=1}^{n} c_i x_i=0\right\}.
$$
\subsection{The Trace Map and the Trace Dual Basis}

The {\it trace map}, $\tr_{q,m}\colon
\efqm\to \efq$, is defined by
$$
\tr_{q,m}(x)\deff x+x^q+x^{q^2}+\cdots+x^{q^{m-1}}.
$$
For ease of notation, we denote the trace map by $\tr$, when $q$ and $m$ are clear from the context.
For a basis of $\efqm$ over $\efq$, $S=\{b_1,\ldots,b_m\}$, the {\it{trace dual basis}} of $S$, $S'=\{b_1',\ldots,b_m'\}$, is a basis of $\efqm$ over $\efq$ for which $\tr(b_i'b_j)=1$ if $i=j$ and $\tr(b_i'b_j)=0$ otherwise. Note that, for every basis there exists a unique trace dual basis.
For $x\in \efqm$ with $x=\sum_{i=1}^mx_i b_i$, $x_i\in \efq$, we have that $x_i=\tr(x b'_i)$, $1\leq i\leq m$.


Let $B=\{b_1,b_2,\ldots,b_r\}\subseteq S$ and let $V=\linspan_{q}(B)$.
The {\it trace-orthogonal} subspace of $V$, $V^{\perp}$, is defined by
$$
V^{\perp}\deff\left\{x\in \efqm~:~\forall v\in V,\ \tr(vx)=0\right\}.
$$
Notice that $\{b_{r+1}',\ldots, b_{m}'\}\subseteq  S$ is a basis for $V^{\perp}$.

\begin{lemma}\label{lemma:rightmult}
Let $S=\{b_1,b_2,\ldots, b_{m}\}$ be a basis of $\efqm$ over $\efq$, let $B\subseteq S$, and let
$V=\linspan_q(S\setminus B)$. For $u\in \efqm$, $w\in V^{\perp}$, and $x\in \efqm$ we have that $x=u\cdot w$ if and only if
$$
\mathbf{w} \cdot [u]_{B,S}= \mathbf{x},
$$
where $\mathbf{w}$ is the vector representation of $w$ according to the basis $B'$ and $\mathbf{x}$ is the vector representation of $x$ according to the basis $S'$.
\end{lemma}

\begin{proof}
Let $[u]_S=\left(u_{ij}\right)$. Then for all $1\leq j\leq m$, the $j$th column of $[u]_S$
is the vector representation of $b_j\cdot u$ according to
the basis $S$. Hence, for all $1\leq i\leq m$, $u_{i,j}=Tr(b'_i\cdot b_j\cdot u)$.
Thus, $[u]_S^T=[u]_{S'}$.

Let $\hat{\mathbf{w}}$ be the vector representation of $w$ according to the basis $S'$. 
We have that
$$
\hat{\mathbf{w}}[u]_S=\mathbf{x}~\Leftrightarrow~~[u]_S\hat{\mathbf{w}}^T=\mathbf{x}^T~~\Leftrightarrow~~x=u\cdot w.
$$
Since $w\in V^{\perp}$ it follows that 
$\hat{\mathbf{w}}$ as zero entries in indices corresponding to elements of $S'\setminus B'$ and hence also in indices corresponding to elements of $S\setminus B$.
Thus,
$$
\hat{\mathbf{w}}[u]_S=\mathbf{x}~~\Leftrightarrow~~\mathbf{w}[u]_{B,S}=\mathbf{x},
$$
which concludes the proof.
\end{proof}

Denote by $\hom_{q}(\efqm,\efq)$ the set of all $\efq$-linear functionals from $\efqm$ to $\efq$. The set $\hom_{q}(\efqm,\efq)$ is an $\efq$-linear space. It is well known that $\hom_{q}(\efqm,\efq)$ is isomorphic to $\efqm$. More precisely, any linear functional
$\efqm\to \efq$ is of the form $y\mapsto \tr(xy)$ for a unique $x\in \efqm$.

Finally, denote by $\sigma:\efqm\rightarrow \efqm$ the Frobenius map, defined by $\sigma(x)=x^q$. Notice that $\sigma$ is an $\efq$-linear map.


\subsection{Reed--Solomon Codes}
Let $A=\{a_1,\ldots,a_n\}\subseteq \efq$ be a subset of $n$
elements, $1\leq n\leq q$, and let
$k\leq n$ be a positive integer. The {\it Reed--Solomon (RS)} code,
$\rs(A,k)_{q}$ is defined as
$$
\rs(A,k)_{q} \deff\left\{\left(f(a_1),\ldots,f(a_n)\right)~:~\begin{array}{c}f\in
\efq[X]\\\deg(f)\leq  k-1\end{array}\right\}.
$$
The set $A$ is called the {\it{evaluation set}} of $\rs(A,k)_{q}$ and we say that the code $\rs(A,k)_{q}$ is {\it{evaluated on}} $A$.
The code $\rs(A,k)_q$ is a linear code of length $n$,
dimension $k$, and minimum distance $n-k+1$. Thus, $\rs(A,k)_q$ is
an MDS code and  can correct up to $n-k$ erasures.

For $\mathbf{v}=(v_1,\ldots,v_n)$, $v_i\in \efq \setminus \{0\}$, $1\leq i\leq n$,
the {\it Generalized Reed--Solomon} (GRS) code, $\grs(A,k,\mathbf{v})_{q}$, is defined as
$$
\grs(A,k,\mathbf{v})_{q}\deff\left\{(v_1c_1,\ldots,v_n
c_n)~:~\mathbf{c}\in \rs(A,k)_{q}\right\}.
$$
We refer to the  vector $\mathbf{v}$ as the {\it{GRS scaling vector}} of $\grs(A,k,\mathbf{v})_{q}$.

It is well known that the dual of a GRS code
is yet another GRS code (see, e.g., \cite[Thm.~5.1.6, p.~66]{H}),
$$
\grs(A,k,\mathbf{v})^{\ast} = \grs(A,n-k,\mathbf{v}'),
$$
where $\mathbf{v}'=(v_1',\ldots,v_n')$ is given by
\begin{equation}\label{eq:grsdual}
v_i'=\frac{v_i}{\prod_{j\neq i}(a_i-a_j)},~~~1\leq i\leq n.
\end{equation}

\subsection{Linear Repair Schemes}\label{subsec:linear_repair}
In what follows, we review the definition of a linear repair scheme, and the important result of
Guruswami--Wootters \cite{GW} that provides a criterion to validate a linear repair scheme.
The result of Dau--Milenkovich
(DM)~\cite{DM} on linear repair schemes for RS codes is also given.
In Section~\ref{sec:lrs_rs} we focus only on linear repair scheme of RS codes evaluated on $\efq$-subpaces
of $\efqm$.

For a linear code $C\subset \efqm^n$ and for
$1\leq i\leq n$, an $\efq$-{\it linear repair scheme} for
the $i$th node (coordinate) of codewords in $C$, in which a surviving node transmits at most $r$ $\efq$-symbols, consists of the following.
\begin{enumerate}
\item A set of $\efq$-linear functionals,
$$L=\left\{g_{j,t}\in \hom_{q}(\efqm,\efq)~:~\begin{array}{c}1\leq j\leq n,~j\neq i,\\
1\leq t\leq r\end{array} \right\},
$$
of size $|L|=(n-1)r$.
\item An $\efq$-linear map $f\colon \efq^{|L|}\to \efqm$,
such that for all $(c_1,c_2,\ldots,c_{n})\in C$, we have
\begin{equation}\label{eq:repair}
c_i=f\left(\left(g_{j,t}(c_j)\right)_{\substack{1\leq j\leq n,~j\neq i\\ 1\leq t\leq r  }}\right).
\end{equation}
\end{enumerate}
\begin{remark}
{\rm
It can be easily verified that
if there exists some function $f$ for which~(\ref{eq:repair}) holds,
then there is also an $\efq$-linear map
for which~(\ref{eq:repair}) holds. Hence, there is no loss of
generality in restricting $f$ to be a linear map.
}
\end{remark}

The {\it repair bandwidth}, $b_i$, of the above repair scheme
for node $i$, is defined as $b_i\deff
\log_2(q)\cdot (n-1)r$, which is the total number
of bits transmitted from the helper nodes in
order to repair the erased node $i$.

For a code $C$, the {\it automorphism group} of $C$, $\aut(C)$,
is the set of all permutations $\tau$ of $\{1,\ldots,n\}$, such that
$\tau \cdot C=C$, where for $\mathbf{c}=(c_1,\ldots,c_{n})\in C$, $\tau\cdot
\mathbf{c}\deff (c_{\tau(1)},\ldots,c_{\tau(n)})$.
\footnote{The automorphism group is indeed a group with composition as its group operation.}
The group $\aut(C)$ is called {\it transitive} if for all
$1\leq i,j\leq n$, there exists $\tau\in \aut(C)$ with $\tau(i)=j$.
If $C$ has a transitive automorphism group, then a linear
repair scheme of $C$ for \emph{some} node can be ``permuted'' in order to become a linear repair scheme for \emph{any} node.

In this paper we are interested in linear repair schemes for RS codes evaluated on $\efq$-subspaces.
Henceforth, $U\subseteq \efqm$ is an $\efq$-subspace of dimension $d\leq m$. For a positive integer
$s<d$, we denote by $C(U,s)$ the $RS$ code evaluated on $U$ with codimension $q^s$, i.e.,
$$
C(U,s)\deff\rs(U,q^d-q^s)_{q^m}.
$$

Clearly, $C(U,s)$ is invariant under any permutation that is a translation by an element of $U$, and hence we have the following well-known lemma.
\begin{lemma}\label{lem:transistive}
The code $C(U,s)$ has a transitive automorphism group.
\end{lemma}

%
From Lemma~\ref{lem:transistive} it follows that if $C(U,s)$ has a linear repair scheme for some node $i$, then it has a linear repair scheme for all nodes.

The following theorem by Guruswami--Wootters \cite{GW} plays an important role in
the proof of the Dau--Milenkovich scheme and is also useful for the proof
of the main result of this paper.

\begin{theorem}\label{thm:criterion}
A linear code
$C\subseteq \efqm^n$ has an $\efq$-linear repair scheme for the $i$th node in which every surviving node transmits at most $r$ $\efq$-symbols, if and only if there exist $m$ dual codewords
$\mathbf{u}_{\ell}=(u_{\ell,1},\ldots,u_{\ell,n})\in C^{\ast}$, $1\leq \ell\leq m$, with the following properties.
\begin{enumerate}
\item  $\rank_q(u_{1,j},\ldots,u_{m,j})\leq r$, for all $j\neq i$.
\item $\rank_q(u_{1,i},\ldots,u_{m,i})=m$.
\end{enumerate}

\end{theorem}

\begin{remark}
As observed in \cite{GW}, a
repair scheme for \emph{one} GRS scaling vector is automatically
also a repair scheme for \emph{all} GRS scaling vectors. In detail,
a repair scheme for $\grs(A,k,\mathbf{v})$, can be converted to a
repair scheme for
$\grs(A,k,\mathbf{v}')$ in the following obvious way. When working with
the latter code, each surviving node $j$ multiplies its content by
$v_j/v'_j$, before using the existing repair scheme, and then the
repaired value of the $i$th node is multiplied by $v'_i/v_i$. In
particular, when repairing RS codes, we may assume without loss of generality that the
dual code in the criterion of Theorem \ref{thm:criterion} is also an RS
code. When this sort of argument will be used ahead, we will say
that some relevant vectors are in the dual code \emph{up
to GRS scaling}.
\end{remark}

%
%

As mentioned in the introduction, the main
result of this paper can be viewed as a generalization
of the Dau--Milenkovich
(DM)~\cite{DM} scheme. 
The DM scheme is given in the following theorem.
\begin{theorem}\label{thm:DM}
  For a set $A\subseteq \efqm$ of size $n$, where $q^{s}<n\leq q^m$ the code
  $RS(A,n-q^s)_{q^m}$ has a linear repair scheme in which each surviving node has to transmit $m-s$
$\efq$-symbols for the repair of the erased node.
\end{theorem}

For an $\efq$-linear subspace $U$ of dimension $d$ and for $1\leq s<d$, the result of Dau and Milenkovich given in Theorem~\ref{thm:DM} states that the code $C(U,s)$ has a linear repair scheme in which each helper node transmits at most $r=m-s$ $\efq$-symbols for the repair of the erased node.
The main contribution of this paper is to show that a lower value of $r$ can be used for the same $s$; in fact, $r$ can be as low as
$m(d-s)/d)$ (with some restrictions on the choice of $U$ for the case $q=2$ and $ms=d(m-r)$).

\section{Linear Repair Schemes For RS Codes Evaluated on $\efq$-Linear Subspaces}\label{sec:lrs_rs}
In this section we introduce
some results
that will be useful in Section~\ref{sec:mainres}, where we present and prove the main result of this paper.

The result of Guruswami--Wooters, presented in Theorem~\ref{thm:criterion}, provides a criterion
to determine if a linear code has a linear repair scheme.
For the code $C(U,s)$, the following proposition provides
an equivalent criterion for the existence of a linear repair scheme that will be useful for the proof of our main theorem.

\begin{proposition}\label{prop:linbymat}
Let $V\subset \efqm$ be an $\efq$-linear subspace of dimension $r$, let
$B_1$ be a basis of $V$, and let $B_2$ be a set of $m-r$ vectors, such that
$S=B_1\cup B_2$ is a basis of $\efqm$ over $\efq$. For a basis $\{u_1,u_2,\ldots,u_d\}$ of $U$,
consider the matrix $M\in \efq^{d(m-r)\times m(s+1)}$ defined by
\begin{equation}\label{eq:matM}
M\deff
\left(\begin{array}{c|ccc}
[u_1]_{B_2,S} &[u_1^q]_{B_2,S} & \cdots &

[u_1^{q^s}]_{B_2,S}\\

[u_2]_{B_2,S} & [u_2^q]_{B_2,S} & \cdots &
[u_2^{q^s}]_{B_2,S}\\

\vdots & \vdots & \ddots & \vdots \\

\undermat{m}{[u_d]_{B_2,S}} &
\undermat{ms}{[u_d^q]_{B_2,S} & \cdots &
[u_d^{q^s}]_{B_2,S}} \\
\end{array}\right)
.
\end{equation}
$$
{}
$$

Write $M=(M_1|M_2)$, where $M_1$ consists of the first $m$ columns of
$M$, and $M_2$ consists of the remaining $ms$ columns. If the column space of $M_1$ is contained in the column space of $M_2$, then $C(U,s)$ has an $\efq$-linear repair scheme in which each surviving node has to transmit at most
$r$ $\efq$-symbols.
\end{proposition}

\begin{proof}
Assume that the column space of $M_1$ is contained in the
column space of $M_2$.
We will prove the existence of a linear repair scheme for the node corresponding to evaluation on $0\in U$ and by Lemma~\ref{lem:transistive}, conclude the existence of a linear repair scheme for all nodes.

Let $W=\linspan_{q}(B_2)$. Any linear combination of the columns of $M_1$ over $\efq$ can be interpreted as
a vector of the form
$$((a_0 u_1)_W, (a_0 u_2)_W, \ldots, (a_0 u_d)_W)^T,$$
for some $a_0\in \efqm $ (recall that $x_{W}$ is the projection of $x$ to $W$). Similarly, any linear combination of the columns of $M_2$ over $\efq$ can be interpreted as a vector of the form
$$\left(\left(\sum_{\ell=1}^sa_{\ell}u_1^{q^{\ell}}\right)_W,
\left(\sum_{\ell=1}^sa_{\ell}u_2^{q^{\ell}}\right)_W,\ldots, \left(\sum_{\ell=1}^sa_{\ell}u_d^{q^{\ell}}\right)_W\right)^T,
$$
for some $a_1, a_2,\ldots, a_s\in \efqm$.

Since the column space of $M_1$ is contained in the
column space of $M_2$, it follows that for every $a_0\in \efqm$, there exist $a_1,\ldots,a_s\in \efqm$ such
that $(a_0u)_{W}=-(a_1u^q+\cdots +a_su^{q^s})_{W}$, for all $u\in U$. Equivalently, the polynomial
$$
f(X)= a_0X+a_1X^q+\cdots+ a_sX^{q^s}
$$ satisfies that $f(u)_W=0$, for all $u\in U$, and hence
$f(U)\subseteq V$.

In particular, if we write $S=\{b_1,\ldots,b_m\}$, then for every $1\leq j\leq m$,
there exist $a_{j,1},a_{j,2},\ldots,a_{j,s}$, such that
$$
f_{j}(X)\deff b_jX+a_{j,1}X^{q}+\cdots+a_{j,s}X^{q^s}
$$
maps $U$ to $V$.

For $1\leq j\leq m$, set
$$
g_j(X)\deff f_j(X)/X.
$$
Then for all $1\leq j\leq m$, $\deg(g_j)\leq q^s-1$ and hence, the evaluation of $g_j$ on $U$ is a codeword of $C(U,s)^{\ast}$, $\mathbf{x}_j=(x_{j,u})_{u\in U}$.
Now, for all $u\in U\smallsetminus\{0\}$, we have
$$
\{x_{j,u}\}_{j=1}^m=\left\{f_j(u)/u\right\}_{j=1}^{m}
\subseteq \frac{1}{u}\cdot V,
$$
so that $\rank_q\left(\left\{x_{j,u}\right\}_{j=1}^m\right)\leq \dim(V)=r$.
Moreover, since $x_{j,0}=b_j$, for all $1\leq j\leq m$, we have that
$$
\rank_q(\left\{x_{j,0}\right\}_{j=1}^m)=\rank_q(S)=m.
$$

The proof follows From Theorem~\ref{thm:criterion}.

\end{proof}

For $M_1,M_2$ defined in Proposition~\ref{prop:linbymat}, a sufficient condition that the column space of $M_1$ is contained in the column space of $M_2$
is that the column space of $M_2$ is equal to $\efq^{d(m-r)}$, or equivalently,
$M_2$ is of full rank and $ms\geq d(m-r)$.
\begin{definition}
A pair $(U,V)$ of $\efq$-linear subspaces of dimensions $d$ and $r$, respectively, is called a {\it good} pair, if the corresponding matrix 
$M_2$ is of full rank and $ms\geq d(m-r)$.
\end{definition} 
Notice that, although the matrix $M$ is defined through
a basis $B_1$ for $V$ and some completion of $B_1$ to a basis $S$ for
$\efqm$ over $\efq$, the goodness of the pair $(U,V)$ does not depend on the choice of these bases.


\begin{lemma}\label{lem:basisU} The goodness of the pair $(U,V)$
does not depend on the choice of the basis $\{u_1,u_2,\ldots,u_d\}$ for $U$.
\end{lemma}

\begin{proof}
Let $\{w_1,w_2,\ldots, w_d\}$ be another basis for $U$ and let $A=(a_{i,j})\in \efq^{d\times d}$ be the non-singular
matrix such that
$$
w_i=\sum_{j=1}^da_{i,j}u_j,
$$
for all $1\leq i\leq d$.
Consider the matrix $\tilde{A}=A\otimes I_{m-r}$,
where $I_{m-r}$ is the $(m-r)\times(m-r)$ identity matrix
and the operation $\otimes$ is the tensor product of matrices.
Then $\tilde{A}$ is a $d(m-r)\times d(m-r)$ non-singular
matrix and hence the matrix $\tilde{A}\cdot M_2$ has the same rank as $M_2$.
The proof of the Lemma follows from the fact that
$$
\tilde{A}\cdot M_2=\left(\begin{array}{ccc}
[w_1^q]_{B_2,S} & \cdots &

[w_1^{q^s}]_{B_2,S}\\

 [w_2^q]_{B_2,S} & \cdots &
[w_2^{q^s}]_{B_2,S}\\

\vdots & \ddots & \vdots \\

[w_d^q]_{B_2,S} & \cdots &
[w_d^{q^s}]_{B_2,S} \\
\end{array}\right)
$$

\end{proof}

For $x_1,x_2,\ldots,x_{\ell}\in \efqm$, define
$$
T(x_1,\ldots,x_{\ell};s)\deff\begin{pmatrix}
x_1^q & x_2^q & \cdots &x_{\ell}^q\\
x_1^{q^2}& x_2^{q^2}& \cdots & x_{\ell}^{q^2}\\
\vdots & \vdots & \ddots & \vdots\\
x_1^{q^s}& x_2^{q^s}&\cdots & x_{\ell}^{q^s}
\end{pmatrix}\in \efqm^{s\times{\ell}}.
$$

\begin{proposition}\label{prop:rank}
For every two positive integers $\ell,s$ and for $x_1,\ldots,x_{\ell}\in \efqm$, if $\rho=\rank_q(x_1,\ldots,x_{\ell})$
then the rank of $T(x_1,\ldots,x_{\ell};s)$ (over $\efqm$) is $\min\{s,\rho\}$.
\end{proposition}

\begin{proof}
Let $\{y_1,\ldots,y_{\rho}\}\subseteq\{x_1,\ldots,x_{\ell}\}$ be a
basis for $\linspan_{q}(x_1,\ldots,x_{\ell})$. Since the Frobenius map, $\sigma_q(x)\equiv x^q$, is an
$\efq$-linear map, it follows that all columns of
$T=T(x_1,\ldots,x_{\ell};s)$ are linear combinations
of those of $T_1=T(y_1,\ldots,y_{\rho};s)$. Hence
$\rank(T)=\rank(T_1)$ and it is sufficient to prove that
$\rank(T_1)=\min\{s,\rho\}$.
For this, it is sufficient to consider the case where $s\leq \rho$,
because for $s\geq \rho+1$, $T(y_1,\ldots,y_{\rho};\rho)$ appears in the
first rows of $T_1$ and if $T(y_1,\ldots,y_{\rho};\rho)$ is of full rank than
$\rank(T_1)=\rho$.

For a vector $\mathbf{a}=(a_1,\ldots,a_s)\in\efqm^s$ such that
$\mathbf{a}T_1=0$, consider the polynomial
$$
g(X)=a_1X^q+a_2X^{q^2}+\cdots+a_sX^{q^{s}}.
$$
Then $y_1,y_2,\ldots,y_{\rho}$ are all roots of $g(X)$ and
since $g(X)$ is an $\efq$-linear map, it follows that all elements of
$\linspan_{q}(y_1,\ldots,y_{\rho})$ are roots of $g(X)$.
Let $f(X)\in \efqm[X]$ be the polynomial
$$
f(X)=a_1^{q^{m-1}}X+a_2^{q^{m-1}}X^q+\cdots+ a_s^{q^{m-1}}X^{q^{s-1}}.
$$
Since the Frobenius map is an $\efq$-linear map, and since
$a^{q^m}=a$, for all $a\in\efqm$, it follows that
$f(X)^q=g(X)$. Hence,
all the roots of $g(X)$ are roots of $f(X)$ as well and
$f(X)$ has at least $q^{\rho}$ roots. However, the degree of $f(X)$ is
$q^{s-1}<q^{\rho}$. It follows that $f(X)$ must be the zero
polynomial and
$a_{\ell}^{q^{m-1}}=a_{\ell}=0$, for all $1\leq \ell \leq s$.

We showed that if $s\leq \rho$ then the rows of $T_1$
are linearly independent over $\efqm$, which concludes the proof.
\end{proof}
%


The following two propositions provide useful characterizations of a
good pair of subspaces $(U,V)$.
\begin{proposition}\label{prop:good}
The following conditions are equivalent.
\begin{itemize}
\item[$1)$] The pair $(U,V)$ is good.
\item[$2)$] For every basis $\{u_1,\ldots,u_d\}$ of $U$ and
for all $v_1',\ldots,v_d'\in
V^{\perp}$ for which
\begin{equation}\label{eq:onemorevg1}
T(u_1,u_2,\ldots,u_d;s)\cdot (v_1',v_2',\ldots, v_d')^T
=\mathbf{0}
\end{equation}
we have that  $v_1'=v_2' =\cdots = v_d'=0$.
\item[$3)$] For every basis $B_2'=\{b_1',\ldots,b_{m-r}'\}$ of
$V^{\perp}$ and for all $w_1,\ldots,w_{m-r}\in U$ for which
\begin{equation}\label{eq:sform1}
T(w_1,w_2,\ldots, w_{m-r};s)\cdot
(b_1,b_2',\ldots,b_{m-r}')^T=\mathbf{0}
\end{equation}
we have that $w_1=w_2 =\cdots = w_{m-r}=0$.
\end{itemize}
\end{proposition}

\begin{proof}
Let $\{u_1,u_2,\ldots,u_d\}$ be any basis for $U$. We first prove that conditions $(1)$ and $(2)$ are equivalent.
By definition, the pair $(U,V)$ is good if and only if $M_2$ is of full rank and $ms\geq d(m-r)$. The
latter holds if and only if $\mathbf{x}=\mathbf{0}$ is the only vector in $\efq^{d(m-r)}$ for which $\mathbf{x}M_2=0$.

A vector $\mathbf{x}\in \efq^{d(m-r)}$ can be represented by $d$ chunks of length $m-r$, such that
the $i$th chunk is the vector representation of some element $v'_i\in V^{\perp}$, $1\leq i\leq d$, according to the basis $B_2'$.
By Lemma \ref{lemma:rightmult} we have that $\mathbf{x}M_2=\mathbf{0}$ is equivalent to
$$
\sum_{i=1}^d v'_i u_i^{q^{\ell}}=0,
$$
for all $1\leq \ell\leq s$, which is equivalent to equation~(\ref{eq:onemorevg1}).

Hence, $(U,V)$ is good if and only if for every basis $\{u_1,u_2,\ldots,u_d\}$ of $U$ and for all $v'_1,v'_2,\ldots, v'_d\in V^{\perp}$, equation~(\ref{eq:onemorevg1}) implies
$$
 v'_1=v'_2=\cdots=v'_{d}=0.
$$

Next, we show that conditions $(2)$ and $(3)$ are equivalent.
Let $w_1,w_2,\ldots,w_{m-r}\in U$
and let $A=(a_{i,j})\in \efq^{d\times (m-r)}$ be the
matrix for which
$w_j=\sum_{i=1}^{d}u_ia_{i,j}$, for all $1\leq j\leq m-r$.

Then, for every basis $\{b'_1,b'_2,\ldots, b'_{m-r}\}$ of $V^{\perp}$,
\begin{equation*}
\begin{aligned}
T{}&(w_1,w_2,\ldots,w_{m-r};s)\cdot (b'_1,b'_2,\ldots,b'_{m-r})^T=\\
&T(u_1,u_2,\ldots,u_d;s)\cdot A\cdot(b'_1,b'_2,\ldots,b'_{m-r})^T=\\
&T(u_1,u_2,\ldots,u_d;s)\cdot(v'_1,v'_2,\ldots,v'_d)^T,
\end{aligned}
\end{equation*}
for $v'_1,v'_2,\ldots,v'_{d}\in V$, such that
$v'_i=\sum_{j=1}^{m-r}a_{i,j}b'_j$, $1\leq i\leq d$.

Hence,
$$
T(w_1,w_2,\ldots,w_{m-r};s)\cdot (b'_1,b'_2,\ldots,b'_{m-r})^T=\mathbf{0},
$$
if and only if
$$
T(u_1,u_2,\ldots,u_d;s)\cdot(v'_1,v'_2,\ldots,v'_d)^T=\mathbf{0}.
$$

Therefore, if condition $(2)$ holds and $w_1,w_2,\ldots,w_{m-r}$ satisfy equation~(\ref{eq:sform1}), then $v'_1=v'_2=\cdots= v'_d=0$ and hence
$A$ is the zero matrix. Thus, $w_1=w_2=\cdots=w_{m-r}=0$ and condition $(3)$ holds as well.
Similarly, condition $(3)$ implies condition $(2)$.
\end{proof}


For a positive integer $i$, define $U^{\wedge q^i}\deff\left\{u^{q^i}~:~u\in U\right\}$. Note that,
since the Frobenius map is $\efq$-linear, it follows that $U^{\wedge q^i}$ is an $\efq$-linear subspace as well.
The next proposition, that is useful for deriving explicit code constructions, suggests a duality between a linear repair scheme of $C(U,s)$, in which each surviving node as to transmits $r$ $\efq$-symbols, and a linear repair scheme of $C(V^{\perp},s)$, in which each surviving node as to transmits $m-d$ $\efq$-symbols.
\begin{proposition}\label{prop:duality}
The pair $(U,V)$ is good if and only if $\left(V^{\perp},
\left(U^{\wedge q^{s+1}}\right)^{\perp}\right)$ is good.
\end{proposition}

\begin{proof}

We will show that if $(U,V)$  is good then $\left(V^{\perp},
\left(U^{\wedge q^{s+1}}\right)^{\perp}\right)$ is also good. Similar arguments can be used
to prove the other direction.

Let $\{b'_1,\ldots,b'_{m-r}\}$ be a basis for $V^{\perp}$.
Assume that $(U,V)$ is good and that $(w_1^{q^{s+1}},w_2^{q^{s+1}},\ldots, w_{m-r}^{q^{s+1}})\in U^{\wedge q^{s+1}}$
satisfies
$$
T(b'_1,b'_2,\ldots, b'_{m-r};s)\cdot(w_1^{q^{s+1}},w_2^{q^{s+1}},\ldots, w_{m-r}^{q^{s+1}})^T=\mathbf{0}.
$$
Then, for all $1\leq \ell \leq s$, we have
$$
\sum_{j=1}^{m-r}w_j^{q^{s+1}}{b'_j}^{q^{\ell}} = 0.
$$
Since $t=s+1-\ell$ satisfies that $1\leq t\leq s$, we can rewrite
the equations as
$$
\sum_{j=1}^{m-r}w_j^{q^{s+1}}{b'_j}^{q^{t}} = 0.
$$
Rasing the $t$th equation to the power of $q^{m+\ell-(s+1)}$
and using the fact that $x^{q^m}=x$, for all $x\in \efqm$, we get
that for all $1\leq \ell\leq s$,
$$
\sum_{j=1}^{m-r} w_j^{q^{\ell}}b'_j=0,
$$
or equivalently
$$
T(w_1,w_2,\ldots,w_{m-r};s)\cdot (b'_1,b'_2,\ldots,{b'}_{m-r})^T=\mathbf{0}.
$$
Since $(U,V)$ is good, it follows from condition $(3)$ of Proposition~\ref{prop:good}
that $w_1=w_2=\ldots=w_{m-r}=0$, and hence $w_j^{q^{s+1}}=0$, for all $1\leq j\leq m-r$.
Thus, by condition $(2)$ of Proposition~\ref{prop:good} we have that
$\left(V^{\perp},
\left(U^{\wedge q^{s+1}}\right)^{\perp}\right)$ is good.

\end{proof}

Let $\Omega$ be the set of all vectors in $\efqm^{m-r}$ whose
entries are $\efq$-linearly independent, i.e., for $\mathbf{x}\in\efqm^{m-r}$, $\mathbf{x}\in \Omega$ if and only if $\rank_q(\mathbf{x})=m-r$.
\begin{lemma}\label{lem:omega_size}
$$
|\Omega|>\frac{q^{m(m-r)}(q-1-q^{-r})}{q-1}.
$$
\end{lemma}

\begin{proof}
The size of $\Omega$ is given by
\begin{equation*}
\begin{aligned}
 |\Omega| &= \prod_{j=0}^{m-r-1}(q^m-q^j) \\
   & = q^{m(m-r)}\prod_{j=0}^{m-r-1}(1-q^{-m+j}).
\end{aligned}
\end{equation*}

A straightforward induction on $n$ shows that for all $n$ positive
real numbers $x_1,\ldots,x_n$, we have that $\prod_{j=1}^n (1-x_j)\geq 1-\sum_{j=1}^n x_j$. Hence,
\begin{eqnarray*}
\frac{|\Omega|}{q^{m(m-r)}} & = & \prod_{j=0}^{m-r-1}(1-q^{-(m-j)}) \\
 & \geq & 1-\sum_{j=0}^{m-r-1}q^{-(m-j)}\\
  & = & 1- \sum_{j=r+1}^{m} q^{-j}\\
 & > & 1- \sum_{j=r+1}^{\infty} q^{-j}\\
 & = & 1-\frac{q^{-(r+1)}}{1-q^{-1}}=1-\frac{q^{-r}}{q-1},
\end{eqnarray*}
as required.
\end{proof}

Let
$$
\bad(U)\deff\left\{\mathbf{x}\in \Omega ~:~ \begin{array}{c}\exists
(u_1,\ldots,u_{m-r})\in U^{m-r} \smallsetminus\{\mathbf{0}\},\\
\hbox{s.t. } T(u_1,\ldots,u_{m-r};s)
\cdot \mathbf{x}^{T}=\mathbf{0}^{T}\end{array}\right\}.
$$
For a pair $(U,V)$ of $\efq$-linear subspaces of dimensions $d$ and $r$, respectively,
let $\mathbf{v'}\in \Omega$ be such that $V^{\perp}=\linspan_{q}(\mathbf{v}')$.
It follows from Proposition \ref{prop:good} that $(U,V)$ is good if and only if $\mathbf{v}'\in\Omega \setminus \bad(U)$. In the next section, we will show the existence of a good pair $(U,V)$ under certain conditions. For this purpose, it will be useful to upper bound the size of $\bad(U)$.

\begin{lemma}\label{lem:Badu_size_1}
For $\mathbf{u}=(u_1,\ldots,u_{m-r})\in U^{m-r}\smallsetminus\{\mathbf{0}\}$,
let $\rho=\rank_q(\mathbf{u})$ and
define the set
$$
\bad(\mathbf{u})\deff\left\{\mathbf{x}\in \Omega
~:~T(u_1,\ldots,u_{m-r};s) \cdot \mathbf{x}^{T}=\mathbf{0}^{T}\right\}.
$$
Then the following holds.
\begin{itemize}
\item[$1)$] If $\rho\leq s$ then $\bad(\mathbf{u})=\emptyset$.
\item[$2)$] If $\rho\geq s+1$ then
\begin{equation}
|\bad(\mathbf{u})|<q^{m(m-r-s)}.
\end{equation}
\end{itemize}
\end{lemma}

\begin{proof}
Let $\{w_1,\ldots,w_{\rho}\}$ be a basis for
$\linspan_{q}(\mathbf{u})$. Since the Frobenius map is $\efq$-linear, it follows that there exists a (unique) matrix $N\in
\efq^{\rho\times(m-r)}$ such that
$$
T(u_1,\ldots,u_{m-r};s) = T(w_1,\ldots,w_{\rho};s)\cdot N.
$$

To prove $(1)$, assume that $\rho\leq s$. It follows from
Proposition~\ref{prop:rank} that the rank of $T=T(u_1,\ldots,u_{\rho};s)$ is $\rho$, and
therefore the columns of $T$ are $\efqm$-linearly independent. Hence,
$T\cdot N\cdot\mathbf{x}^{T}=\mathbf{0}^{T}$ if and only if $N\cdot\mathbf{x}^{T}=\mathbf{0}^{T}$.
However, for all $\mathbf{x}\in \Omega$, we have that $N\cdot
\mathbf{x}^{T}\neq \mathbf{0}$. This is true since all entries of $N$ belong to
$\efq$, $N$ is not the zero matrix, and the entries of
$\mathbf{x}$ are $\efq$-linearly independent. Hence, $T(u_1,\ldots,u_{m-r};s) \cdot \mathbf{x}^{T}\neq \mathbf{0}^{T}$, for
all $\mathbf{x}\in \Omega$, which implies that $\bad(\mathbf{u})=\emptyset$.

For the proof of $(2)$, assume that $\rho\geq s+1$ (note that since $\rho\leq m-r$, this
implies in particular that $s\leq m-r-1$). It follows from Proposition
\ref{prop:rank}, that $\rank_{q^m}(T(u_1,\ldots,u_{m-r};s))=s$. Hence,
the $\efqm$-dimension of the right-kernel of $T(u_1,\ldots,u_{m-r};s)$
is $m-r-s$. Thus,
\begin{equation}
|\bad(\mathbf{u})|<q^{m(m-r-s)}.
\end{equation}

\end{proof}

\begin{lemma}\label{lem:bad_size}
Let $a$ be a common factor of $m$
and $d$ and assume that $U\subseteq \efqm$ is an $\ef_{q^a}$-linear subspace of
dimension $d/a$. Then,

\begin{equation}\label{eq:numbad}
|\bad(U)|
< \frac{q^{d(m-r)-ms}}{q^a-1} q^{m(m-r)}
\end{equation}
\end{lemma}

\begin{proof}
Define an
equivalence relation $\sim$ on
$\efqm^{m-r}\smallsetminus\{\mathbf{0}\}$ by setting $\mathbf{x}\sim\mathbf{y}$
if and only if there exists $\beta\in \ef_{q^a}\smallsetminus\{0\}$ such that
$\mathbf{x}=\beta\cdot \mathbf{y}$. Note that, since $U$ is a vector space over
$\ef_{q^a}$, the equivalence class of any $\mathbf{u}\in
U^{m-r}\smallsetminus \{\mathbf{0}\}$ is contained in $U^{m-r}\smallsetminus
\{\mathbf{0}\}$.

Let $\reps\subset U^{m-r}\smallsetminus\{\mathbf{0}\}$ be a set consisting
of a single representative from each
equivalence class of $\sim$ in $U^{m-r}\smallsetminus\{\mathbf{0}\}$, and
note that $|\reps|=(q^{d(m-r)}-1)/(q^a-1)$. Note also that, as
$T(\beta\cdot\mathbf{u};s)=\diag\big(\{\beta^{q^i}\}_{i=1}^s\big)\cdot
T(\mathbf{u};s)$ ($\diag(\beta_1,\ldots,\beta_n)$ is the diagonal $n\times n$ matrix $D$ with $D_{i,i}=\beta_i$, $1\leq i\leq n$), for all $\beta\in \efqm$, $\mathbf{u}\sim\mathbf{u}'$ implies
$\bad(\mathbf{u})=\bad(\mathbf{u}')$. It
follows from the above
comments and from Lemma~\ref{lem:Badu_size_1} that
\begin{equation*}
\begin{split}
|\bad(U)|&=\Big|\bigcup_{\mathbf{u}\in \reps} \bad(\mathbf{u}) \Big|
  \leq  \sum_{\mathbf{u}\in \reps} |\bad(\mathbf{u})| \\
&< \frac{q^{d(m-r)}-1}{q^a-1} q^{m(m-r-s)}
< \frac{q^{d(m-r)-ms}}{q^a-1} q^{m(m-r)}
\end{split}
\end{equation*}

\end{proof}

An intriguing question is wether or not, for all $U$, all pairs $(U,V)$
are good, or equivalently, does $|\bad(U)|=0$. The answer to this question is given in the next proposition.
The proof can be found in the appendix.

\begin{proposition}\label{prop:bad_example}
Assume that $1\leq s<d<m$ and $ms\geq d(m-r)$.
Then for every $\efq$-linear subspace $U\subseteq \efqm$ of dimension $d$, the following holds.
\begin{itemize}
\item[1)] If $r\geq m-s$ then $|\bad(U)|=0$.
\item[2)] If $r<m-s$ then $|\bad(U)|>0$.
\end{itemize}
\end{proposition}
\section{Existence of Linear Repair Schemes }\label{sec:mainres}
In this section we present and prove the main result of the paper, namely, the existence of a linear repair scheme for $C(U,s)$, in which surviving nodes transmit at most $r$ symbols from $\efq$.
\begin{theorem}\label{thm:main}
The code $C(U,s)$ has an $\efq$-linear repair scheme
in which each surviving node transmits
$r$ $\efq$-symbols, provided that one of
the following conditions holds.

\begin{enumerate}

\item $q\geq 3$ and $ms\geq d(m-r)$.
\item $q=2$, $r\geq 2$, and $ms \geq d(m-r) + 1$.
\item $q=2$, $ms = d(m-r)$ and $U$ is a
$\ef_{q^a}$-linear subspace of $\efqm$ of dimension $d/a$, for $a=\gcd(m,d)$.

\end{enumerate}
\end{theorem}

Notice that, the third condition of Theorem~\ref{thm:main} includes a more strict restriction on $U$, i.e, $U$ is required to be an $\ef_{q^a}$-subspace of $\efqm$ of dimension $d/a$. This requirement on $U$ is stronger, since any such subspace of $\efqm$ is also an $\ef_q$-subspace of dimension $d$. In addition, if $d$ and $m$ are co-prime, i.e., $a=1$, the equality $s=d(m-r)/m$ implies that $r=m$ and $s=0$, and hence $C(U,s)$ is an RS code of length $n=q^d$ and dimensions $k=n-1$. This special case trivially holds, since such a code can correct any node failure when all surviving nodes transmit their entire content.

The proof of Theorem~\ref{thm:main} involves a probabilistic
argument in which an $\efq$-subspace $V\subset \efqm$ of dimension $r$ is chosen uniformly at random. If the pair $(U,V)$ is good then by Proposition~\ref{prop:linbymat} a linear repair scheme for the code $C(U,s)$ is guaranteed. Moreover, the goodness of the pair $(U,V)$ can be verified, using Gaussian elimination, in polynomial time. We will show that
the probability that $(U,V)$ is good is fairly large (at least $1/3$) and thus obtain a practical probabilistic algorithm to
construct the promised repair scheme for each subspace $U$
guaranteed by Theorem~\ref{thm:main}.


%
%
%
%
%
%

In what follows, we assume that $\mathbf{v}'=(v_1',\ldots,v_{m-r}')$ is a vector drawn uniformly at random from the set $\Omega$, i.e., $\mathbf{v}'\in \efqm^{m-r}$ and $\rank_q(\mathbf{v}')=m-r$.
The proof of Theorem~\ref{thm:main} will follow immediately from the next theorem and corollary.
\begin{theorem}\label{thm:main2}
Let $\mathbf{v'}\in \Omega$ and let $V\subset \efqm$ be the $\efq$-linear subspace of dimension $r$ such that
$V^{\perp}=\linspan_{q}(\mathbf{v}')$.
For a positive integer $a$, if $a$ is a common factor of $m$
and $d$, and $U\subseteq \efqm$ is an $\ef_{q^a}$-subspace of
dimension $d/a$, the probability that $(U,V)$ is good is at least
\begin{equation}\label{eq:prob}
1-\frac{q^{d(m-r)-ms}}{q^a-1}\cdot\frac{q-1}{q-1-q^{-r}}.
\end{equation}
\end{theorem}

\begin{proof}

We will use a counting argument based on Proposition~\ref{prop:good}.
By Proposition \ref{prop:good}, $(U,V)$ is good if and only if $\mathbf{v}'\in\Omega \setminus \bad(U)$.
Hence,
\begin{equation}\label{eq:rawprob}
\prob\big((U,V)\text{ is good};U\big) = 1-\frac{|\bad(U)|}{|\Omega|}.
\end{equation}

Combining~(\ref{eq:rawprob}) with Lemmas~\ref{lem:omega_size} and~\ref{lem:bad_size} we have that
$$
\prob\big((U,V)\text{ is good}\big)> 1- \frac{q^{d(m-r)-ms}}{q^a-1}\cdot\frac{q-1}{q-1-q^{-r}}.
$$

\end{proof}

\begin{corollary}\label{cor:main}
If $U\subseteq\efqm$ is an $\efq$-subspace
of dimension $d$ and $p$ is the probability that
$(U,V)$ is good, then the following statements hold.
\begin{enumerate}
\item If $q\geq 3$ and $ms\geq d(m-r)$ then $p\geq 2/5$.
\item If $q=2$, $r\geq 2$, and $ms\geq d(m-r)+1$ then $p\geq 1/3$.
\item Let $a=\gcd(m,d)$. If $q=2$, $a\geq 2$, $ms = d(m-r)$, and $U$ is also an $\ef_{q^a}$-subspace of $\efqm$
of dimension
$d/a$, then $p\geq 1/3$.
\end{enumerate}
\end{corollary}

\begin{proof}
Let $h$ be the right hand side of~(\ref{eq:prob}).
Then $h$ is minimized when $ms-d(m-r)$, $r$, $q$, and $a$ are minimized.
If the conditions of $(1)$ hold, then the minimum of $h$ is obtained for $q=3$, $r=1$, $ms=d(m-r)$, and $a=1$ and is equal to $2/5$.
If the conditions of $(2)$ hold, then the minimum of $h$ is obtained for $r=2$, $ms=d(m-r)+1$, and $a=1$ and is equal to $1/3$.
Lastly, for the conditions of $(3)$, the minimum of $1/3$ is obtained for $a=2$ and $r=1$.
\end{proof}

\section{Explicit Constructions}\label{sec:construction}
In this section we present explicit constructions of linear repair schemes for $C(U,s)$, for a specific choice of the $\efq$-linear subspace $U$, where $(m-r)$ divides $m$ or $d$ divides $m$.

First, we present a construction for the code $C(U,s)$, for some $\efq$-linear subspace $U\subseteq \efqm$,
where $m-r$ divides $m$.
Recall that by Proposition~\ref{prop:linbymat},
it is suffices to show an explicit choice of an $\efq$-linear
subspace $V$ of $\efqm$ of dimension $r$ such that the pair $(U,V)$ is good.
\begin{proposition}\label{prop:mrm22}

Assume that $(m-r)$ divides $m$, $d<m$, and $ms\geq d(m-r)$. Let $\alpha\in \efqm$ be a primitive element and let
$U\deff\linspan_{q}(1,\alpha,\ldots,\alpha^{d-1})$.
Then the pair $(U,V=\ef_{q^{m-r}}^{\perp})$
is good.
\end{proposition}

\begin{proof}
First, notice that since $m-r$ divides $m$, it follows that $\ef_{q^{m-r}}$ is a subfield of $\ef_{q^m}$ and therefore, $V$ is a well defined $\ef_{q}$-subspace of $\ef_{q^m}$ of dimension $r$.
It suffices to prove the case where $s$ takes its minimum possible value
$\lceil d(m-r)/m \rceil$, since if we show that $M_2$ from
Proposition \ref{prop:linbymat} is of full rank for the minimal $s$ then it also holds for larger values of $s$. In particular, we may assume that
$s\leq m-r$.

By condition $(2)$ of Proposition \ref{prop:good}
it is sufficient to prove that for all $f\in\ef_{q^{m-r}}[X]$ with $\deg(f)<d$, if $f(\alpha^q)=f(\alpha^{q^2})=\cdots=f(\alpha^{q^s})=0$ then $f$ must be the zero polynomial.

If $x\in\efqm$ is a root of $f(X)$, then so are the
conjugates $x^{q^{j(m-r)}}$, for all $1\leq j< m/(m-r)$. Hence, it is
sufficient to prove that if all elements of
\begin{equation}\label{eq:R}
R\deff\left\{\alpha^{q^{j(m-r)+i}}~:~1\leq i\leq s,
0\leq j<m/(m-r)\right\}
\end{equation}
are roots of $f$, then $f$ is the zero polynomial.

Since $s\leq m-r$, all exponents $j(m-r)+i$ of $q$ appearing in
(\ref{eq:R}) are positive, distinct, and smaller than $m+1$.
It follows that $|R|=sm/(m-r)\geq d > \deg(f)$, and thus $f$ must be the zero polynomial, as
required.
\end{proof}
By Proposition~\ref{prop:mrm22} and by the duality of the goodness property given in Proposition~\ref{prop:duality}, 
we conclude that if $d$ divides $m-r$ then the pair $(\ef_{q^{d}},\tilde{U})$ is good, for some $\efq$-linear subspace $\tilde{U}\subseteq q^m$ of dimension $m-r$ that can be derived from $\linspan_{q}(1,\alpha,\ldots,\alpha^{m-r-1})$. Thus, we have an explicit construction of a linear repair scheme for $C(\ef_{q^d},s)$, in which each surviving node has to transmit at most $r$ $\efq$-symbols for the rapier of the erased node.
This construction is also a straightforward generalization of the
DM scheme and is similar to the ``scheme in one coset'' proposed by Li {\it et al.}~\cite{LWJ2}. The result is summarized in the next proposition, to which we present an alternative proof that is
based on a simple but useful argument.

\begin{proposition}\label{prop:mrm}
Assume that $d$ divides $m$, $d<m$, and that $ms=d(m-r)$.
Then the code $C(\ef_{q^d},s)$ has a linear repair scheme in which
each surviving node has to transmit at most $r$
$\ef_q$-symbols for the repair of the erased node.
\end{proposition}

\begin{proof}
Let $\{b_1,\ldots,b_{m/d}\}$ be any basis for $\efqm$ over $\ef_{q^d}$, and let
$\{b'_1,\ldots,b'_{m/d}\}$ be its dual basis. For a polynomial $f(X)\in\efqm[X]$
of degree at most $k-1$, where $k=q^d-q^s$, there exist polynomials $f_j(X)\in \ef_{q^d}[X]$ of degree at most $k-1$ such that
$$
f(X)=b_1f_1(X)+\cdots+b_{m/d}f_{m/d}(X).
$$
Hence, the codeword $\mathbf{c}=(f(\alpha))_{\alpha\in \ef_q^d}$
can be represented by the $m/d$ codewords of
$\rs(\ef_{q^d},q^d-q^s)_{\ef_{q^d}}$, $\mathbf{c}_j=(f_j(\alpha))_{\alpha\in \ef_q^d}$.
In addition, for $\beta\in
\ef_{q^d}$ and for all
$1\leq j\leq m/d$,
\begin{equation}\label{eq:fjbeta}
f_j(\beta)=\tr_{q^d, m/d} (f(\beta)\cdot b'_j).
\end{equation}
This implies that a linear repair scheme of $\rs(\ef_{q^d},q^d-q^s)_{\ef_{q^d}}$,
in which each surviving node transmits at most $r'$ $\efq$-symbols, results in a linear repair scheme for
$C(\ef_{q^d},s)$ in which each surviving node transmits at most $r=r'm/d$ $\efq$-symbols.

By the DM scheme, for every $1\leq s<d$, $\rs(\ef_{q^d},q^d-q^s)_{\ef_{q^d}}$ has a linear repair scheme in which each surviving node has to transmits at most $r'=d-s$ symbols, which concludes the proof.
\end{proof}

\section{Examples}\label{sec:eg}

In Table~\ref{table:14_10} we consider two specific examples of linear codes with
linear repair schemes
that are obtained from our constructions
and compare their bandwidth to known linear repair schemes of these codes.

We first consider the well known $[14,10]_{\ef_{2^8}}$ GRS code deployed at the Facebook
Hadoop Analytic cluster (see, e.g.,
\cite[Sec.~V.C]{GW} and references therein). Using Proposition~\ref{prop:mrm}, we construct $C(U,s)$ code over $\ef_{2^8}$ with $U=\ef_{2^4}$, $s=2$ and $r=4$. The code $C(U,s)$ is a $[16,12]_{2^8}$ code. We then shorten this code to obtain a $[14,10]_{\ef_{2^8}}$ code with a linear repair scheme in which $r=4$ and the bandwidth is $b=52$. This construction was also given in~\cite{LWJ2}. A naive decoding of an RS code over with $\ef_{2^8}$ with dimension 10 has bandwidth $80$, while the linear repair scheme from~\cite{DM} achieves a bandwidth of $54$, where not all surviving node transmitting the same number of bits.

The second code we consider is $C(U,s=4)$, where $U$ is an $\ef_{2^3}$-subspace of
$\ef_{2^{15}}$ of dimension two. Hence, $U$ is an $\ef_{2}$-subspace of dimension $6$ and from Theorem~\ref{thm:main}, $C(U,s)$ has a linear repair scheme in which $r=5$.
This code is a $[64,48]_{2^{15}}$ RS code. The bandwidth of a naive approach and the main scheme from~\cite{DM} are presented in Table~\ref{table:14_10}.

\begin{table}
\begin{center}
\begin{tabular}{|c|c|c|c|c|c|c|}
  \hline
   Repair Scheme   &$q$ & $m$ & $n$ & $k$ & $r$ & $b$ \\
\hline
Prop.~\ref{prop:mrm}& 2 &  8 & 14 & 10 & 4 & 52 \\
\hline
Naive& 2 & 8 & 14 & 10 & 8 & 80 \\
  \hline
DM& 2 & 8 & 14 & 10 & - & 54 \\
\hline
\hline
Thm.~\ref{thm:main}& 2 & 15 & 64 & 48 & 5 & 315 \\
\hline
Naive& 2 & 15 & 64 & 48 &  15 & 720 \\
\hline
DM & 2 & 15 & 64 & 48 & 11 & 693\\

\hline

\end{tabular}\vspace{0.2cm}\caption{}\label{table:14_10}
\vspace{-1.0cm}
\end{center}
\end{table}
Lastly, we consider the case $q=2$, $r\geq 2$, $s=1$, and $ms=d(m-r)$, where $\gcd(m,d)>1$.
A linear repair scheme for these parameters is guaranteed by Theorem~\ref{thm:main}. The constructed RS codes have
two parity symbols. Since $m=d(m-r)$, it follows that $r=m(d-1)/d$ and
the bandwidth is $(n-1)m(d-1)/d$, where $n=q^d$.
A construction of linear repair schemes for
RS codes of codimension $2$ over $\ef_{2^m}$ is also
given in \cite[Thm.~10]{GW}, with repair bandwidth $3(n-1)m/4$,
where $n\leq 2^{m/2+1}$ is the length of the code.
This shows that in general, the bandwidth of the
scheme of Theorem \ref{thm:main2} is not minimal. Note that, the scheme
of \cite[Thm.~10]{GW} is imbalanced, in the sense that about half
of the surviving nodes transmit half of their content, while the
remaining surviving nodes transmit their entire content. Moreover, the
evaluation set in this scheme is not a linear subspace.

\section{Conclusion}\label{sec:conclusion}

In this work we studied the repair problem for RS
codes, evaluated on an $\efq$-linear subspace $U\subseteq\efqm$ of dimension $d$.
For this class of RS codes, we showed the existence of linear
repair schemes, in which each surviving node transmits at most $r$ $\efq$-symbols for the repair of the erased node,
for a wide range of parameters. This result relies on the existence of an $\efq$-linear
subspace $V\subseteq\efqm$ of dimension $r$
for which the pair $(U,V)$ is good. It also yields
a practical probabilistic construction of a linear repair scheme.
We also showed that if $r<m-s$, where $q^s$ is the codimension of the RS code,
and if $V$ is chosen uniformly at random,
then the probability that $(U,V)$ is good is strictly less than one. Thus, in this case, the probabilistic construction is not trivial in the sense that not every pair $(U,V)$ is good.
Our results expand the Dau--Milenkovich scheme and one of the schemes of Li et al., for a wide range of parameters, where $r<m-s$.

Another contribution of this paper is that the presented scheme as a duality property in the following sense;
A good pair $(U,V)$ of $\efq$-linear subspaces of dimensions $d$ and $r$ can be used to construct a good pair of
$\efq$-linear subspaces of dimensions $m-r$ and $m-d$, $\left(V^{\perp},\tilde{U}\right )$, where $\tilde{U}=\left(U^{\wedge q^{s+1}}\right)^{\perp}$.
This duality property is useful for explicit constructions.

For a wide range of parameters, our scheme provides RS codes of codimension $q^s$, where the minimal value of $s$ is $d(m-r)/m$.
For future research, it will be interesting to understand if the this scheme is optimal for RS codes evaluated on linear subspaces.


\appendix

The purpose of this Appendix is to prove Proposition~\ref{prop:bad_example}. That is,
to show that
if $1\leq s<d<m$ and $ms\geq d(m-r)$, then for every $\efq$-linear subspace $U\subseteq \efqm$ of dimension $d$, there exists an $\efq$-linear subspace $V\subseteq \efqm$ of dimension $r$ for which the pair $(U,V)$ is not good if and only if $r<m-s$.
In fact, we prove a somewhat stronger result, namely that for $r<m-s$, there exists a pair $(U,V)$ that
is not good even in the weaker sense, as the corresponding matrices
$M_1,M_2$ defined in Proposition~\ref{prop:linbymat} satisfy that the column space of $M_1$
does not contained in the column space of $M_2$.
Similarly, if $r\geq m-s$, every pair $(U,V)$ is good in the weaker sense.

We first prove the first part of Proposition~\ref{prop:bad_example}, that is, the case $r\geq m-s$.
The proof follows the lines of the proof of the DM scheme stated in Theorem~\ref{thm:DM}.

\renewcommand*{\proofname}{Proof of Proposition~\ref{prop:bad_example} Part (1)}
\begin{proof}
First notice that, it is sufficient to prove the claim for $r=m-s$, since every $\efq$-linear
subspace $V$ of dimension $r>m-s$ contains an $\efq$-linear subspace, $W$,
of dimension $m-s$, and if $(U,W)$ is good then $(U,V)$ is good (this holds even in the weaker sense).

As shown in the proof of Proposition~\ref{prop:linbymat}, if
the column space of $M_1$ is contained in
the column space of $M_2$ (over $\efq$), then for every $a_0\in \efqm$, there exist $a_1,a_2,\ldots,a_s\in \efqm$, such that the
$\efq$-linearized polynomial $f(X)=a_0X+a_1X^q+\cdots+a_sX^{q^s}$ maps $U$ to $V$. The other direction also holds.

Let $V$ be an $\efq$-subspace of dimension $r=m-s$.
The image polynomial of $V$, $\fim_V(X)=b_0X+b_1X^q+\cdots +b_sX^{q^s}\in \efqm[X]$
is an $\efq$-linearized polynomial of degree $q^s$ that maps
$\efqm$ onto $V$. In particular $\fim_V(U)\subseteq V$. Notice that,
there exists a unique image polynomial of $V$, for all $V$
(see \cite{BK09} and the references therein).
The kernel of $\fim_V(X)$ is an $\efq$-linear subspace of $\efqm$ of
dimension $s$. This implies that all the roots of $\fim_V(X)$ are distinct, i.e., $\fim_V(X)$ is separable, and hence $b_0\neq 0$.

Now, for $a_0\in \efqm$, let $f_{a_0}(X)=\fim_V(a_0\cdot b_0^{-1}X)$. Then $f_{a_0}(X)=a_0X+a_1X^q+\cdots+a_sX^{q^s}$, for some $a_1,\ldots, a_s\in \efqm$, and $f_{a_0}(X)$ maps $\efqm$ to $V$. In particular, $f_{a_0}(U)\subseteq V$.
\end{proof}

Next, we prove the second part of Proposition~\ref{prop:bad_example},
namely, the case $r<m-s$.
But first, we need the following lemma.
\renewcommand*{\proofname}{Proof}
\begin{lemma}\label{lem:weak_good}
For a pair $(U,V)$ of $\efq$-linear subspaces of $\efqm$ of dimensions $d$ and $r$, respectively,
and for the corresponding matrices $M_1$ and $M_2$ as defined in Proposition~\ref{prop:linbymat}, the following are equivalent.
\begin{itemize}
\item[$1)$] The column space of $M_1$ is contained in the column space of $M_2$.
\item[$2)$] For every basis $\{u_1,\ldots,u_d\}$ of $U$ and
for all $v_1',\ldots,v_d'\in
V^{\perp}$ such that,
\begin{equation}\label{eq:onemorevg}
T(u_1,u_2,\ldots,u_d;s)\cdot (v_1',v_2',\ldots, v_d')^T
=\mathbf{0}
\end{equation}
we have that  $\sum_{i=1}^du_iv'_i=0$.
\item[$3)$] For every basis $B_2'=\{b_1',\ldots,b_{m-r}'\}$ of
$V^{\perp}$ and for all $w_1,\ldots,w_{m-r}\in U$, such that
\begin{equation}\label{eq:sform}
T(w_1,w_2,\ldots, w_{m-r};s)\cdot
(b_1,b_2',\ldots,b_{m-r}')^T=\mathbf{0}
\end{equation}
we have that $\sum_{j=1}^{m-r}w_jb'_j=0$.
\end{itemize}
\end{lemma}

\begin{proof}
We first prove that conditions $(1)$ and $(2)$ are equivalent.
Let $u_1,u_2,\ldots,u_d$ be a basis for $U$.
Note that, the column space of $M_1$ in contained in the column space of $M_2$ if and only if
the left kernel of $M_2$ is contained in the left kernel of $M_1$.
Equivalently, for all $\mathbf{x}\in \efq^{d(m-r)}$ for which
$\mathbf{x}M_2=\mathbf{0}$, we have that
$\mathbf{x}M_1=\mathbf{0}$.

The proof proceeds along the lines of the proof of Proposition~\ref{prop:good}, by representing the equations
$\mathbf{x}M_2=\mathbf{0}$ and $\mathbf{x}M_1=\mathbf{0}$ as equation~(\ref{eq:onemorevg}) and $\sum_{i=1}^dv'_iu_i=0$, respectively.

To prove that conditions $(2)$ and $(3)$ are equivalent, we follow the lines of the corresponding part of the proof of Proposition~\ref{prop:good}.

\end{proof}

\renewcommand*{\proofname}{Proof of Proposition~\ref{prop:bad_example} Part (2)}
\begin{proof}
By Lemma~\ref{lem:weak_good}, given an $\efq$-linear subspace $U\subseteq \efqm$
of dimension $d$, we need to show the existence of an $\efq$-subspace $V\subseteq \efqm$ of
dimension $r$, such that for some $w_1,w_2,\ldots,w_{m-r}\in U$ and for some basis $B_2'=\{b'_1,
b'_2,\ldots,b'_{m-r}\}$ of $V^{\perp}$ we have that
$$
T(w_1,\ldots,w_{m-r};s)\cdot(b'_1,
b'_2,\ldots,b'_{m-r})^T=\mathbf{0}
$$
and
$$
\sum_{j=1}^{m-r}w_jb'_j\neq 0.
$$

Let $\tau=\min\{d,m-r\}$ and let $w_1,\ldots,w_{\tau}\in U$ be any $\efq$-linearly independent elements.
Consider the matrix $T_1=T(w_1,\ldots,w_{\tau};\tau-1)$. By Proposition~\ref{prop:rank},
we have that the rank of $T_1$ is $\tau-1$. Hence there exist $b'_1,b'_2,\ldots,b'_{\tau}\in \efqm$, not all zeros,
such that $T_1\cdot (b'_1,b'_2,\ldots, b'_{\tau})^T=\mathbf{0}$.

Next, we will show that $b'_1,\ldots, b'_{\tau}$ are $\efq$-linearly independent.
Assume to the contrary that $b'_{\tau}=\sum_{i=1}^{\tau-1}a_ib'_i$, for some $a_1,\ldots,a_{\tau-1}\in \efq$.
Then,
\begin{equation*}
\begin{aligned}
\mathbf{0}&=T_1\cdot (b'_1,b'_2,\ldots,b'_{\tau})^T\\
&=T_1\cdot(b'_1,b'_2,\ldots,b'_{\tau-1},\sum_{i=1}^{\tau-1}a_ib'_i)^T\\
&=T(w_1,w_2,\ldots,w_{\tau-1};\tau-1)\cdot(b'_1,\ldots,b'_{\tau-1})^T\\
&+T(w_{\tau};\tau-1)(a_1b'_1,a_2b'_2,\ldots,a_{\tau-1}b'_{\tau-1})^T\\
&=T(w'_1,w'_2,\ldots,w'_{\tau-1};\tau-1)\cdot(b'_1,\ldots,b'_{\tau-1})^T,
\end{aligned}
\end{equation*}
where $w'_i=w_i+a_iw_{\tau}$, $1\leq i\leq \tau-1$. Since $w_1,\ldots, w_{\tau}$ are $\efq$-linearly
independent, it follows that $w'_1,\ldots, w'_{\tau-1}$ are also $\efq$-linearly independent. By Proposition~\ref{prop:rank}
we have that
$T(w'_1,w'_2,\ldots,w'_{\tau-1};\tau-1)$ is non-singular, hence $b'_1,\ldots,b'_{\tau}$ must all be zeros and we derived a contradiction.

Define $w_{\tau+1}=\cdots =w_{m-r}=0$ and choose $b'_{\tau+1},\ldots,b'_{m-r}$ such that
$B_2'=\{b'_1,b'_2,\ldots,b'_{m-r}\}$ is a basis for some $\efq$-linear subspace $V^{\perp}$
of dimension $m-r$.

Then, since $r<m-s$, it follows that $s<m-r$, and hence, recalling that $s<d$, we have $s<\tau$.
Thus,
\begin{equation*}
\begin{aligned}
T{}&(w_1,w_2,\ldots,w_{m-r};s)\cdot (b'_1,b'_2,\ldots,b'_{m-r})^T=\\
&T(w_1,w_2,\ldots,w_{\tau};s)\cdot(b'_1,b'_2,\ldots,b'_{\tau})=\mathbf{0}.
\end{aligned}
\end{equation*}

Finally, we need to show that $\sum_{i=1}^{\tau}w_ib'_i\neq 0$.
Let $\sigma^{-1}:\efqm\rightarrow \efqm$ be the inverse of the Frobenius map, $\sigma:\efqm\rightarrow \efqm$ ($\sigma(x)=x^q$, for all $x\in \efqm$).
Let $z_i=\sigma^{-1}(w_i)$, for $1\leq i\leq \tau$. We have that $z_1,z_2,\ldots,z_{\tau}$ are $\efq$-linearly independent and hence, by
Proposition~\ref{prop:rank},
$T(z_1,z_2,\ldots,z_{\tau};\tau)$ is non-singular over $\efqm$. Thus,
$T(z_1,z_2,\ldots,z_{\tau};\tau)\cdot (b'_1,b'_2,\ldots,b'_{\tau})^T\neq \mathbf{0}$.
Writing the equations, we have that for some $0\leq \ell\leq \tau-1$,
$$
\sum_{i=1}^{\tau}{w_i}^{q^{\ell}}b'_i\neq 0.
$$
However, $b'_1,\ldots,b'_{\tau}$ satisfy that for all $1\leq \ell\leq \tau-1$,
$$
\sum_{i=1}^{\tau}{w_i}^{q^{\ell}}b'_i=0,
$$
and thus $\sum_{i=1}^{\tau}w_ib'_i\neq 0$, which concludes the proof.
\end{proof}



\end{document}